\documentclass[12pt]{article}
\usepackage[margin=2.2cm]{geometry}
\usepackage[T1]{fontenc}
\usepackage[english]{babel}
\usepackage{amsmath,amssymb,latexsym, booktabs}
\usepackage{graphicx}
\usepackage{pifont}
\usepackage{verbatim}
\usepackage{url}
\usepackage{tikz}

\newtheorem{theorem}{Theorem}[section]
\newtheorem{definition}[theorem]{Definition}

\newtheorem{proposition}[theorem]{Proposition}

\newtheorem{corollary}[theorem]{Corollary}

\def\altbox{\hspace{2mm}\nolinebreak\null\nolinebreak\hfill\Box}
\newenvironment{proof}{\noindent {\bf Proof:}}{$\altbox$\bigskip}

\usetikzlibrary{arrows,shapes,automata}
\tikzstyle{blueVertex}=[circle,draw=blue!50,fill=blue!20,thick]
\tikzstyle{darkBlueVertex}=[circle,draw=blue!75,fill=blue!55,thick]
\tikzstyle{block} = [rectangle, draw=black!60, fill=blue!20, top color=green!60,
    text width=5em, text centered, minimum height=4em]
\tikzstyle{redVertex}=[circle,draw=red!60,fill=red!30,thick]
\tikzstyle{inVertex} = [ ]
\tikzstyle{edge} = [draw,thick,->]
\tikzstyle{edgeDot} = [draw,dashed,->]
\tikzstyle{selectedEdge} = [draw,line width=5pt,-,blue!50]
\tikzstyle{redEdge} = [draw,line width=5pt,-,red!50]
\tikzstyle{bigEdge} = [draw,line width=5pt,->,red!75,>=triangle 60] 
\tikzstyle{bigBlueEdge} = [draw,line width=2pt,->,blue!75,>=triangle 60] 
\tikzstyle{weight} = [font=\small]
\tikzstyle{block} = [rectangle, draw, fill=blue!20, text width=5em, text
centered, rounded corners, minimum height=4em]
\newcommand{\At}{\A_{\tau}}
\newcommand{\gor}{\;\;\big|\;\;}
\newcommand{\Nilp}{{0}}
\newcommand{\A}{\mbox{$\mathbb{A}$}}
\newcommand{\nat}{\mathbb{N}}
\newcommand{\Real}{\mathbb{R}}

\newcommand{\name}[1]{\mbox{\sc #1}}
\newcommand{\sos}[2]{%
\def\arraystretch{1,3}\begin{array}{c}#1\\\hline #2\end{array}%
}
\newcommand{\nar}[2]{\xrightarrow{#1}_{#2}}
\newcommand{\Nar}[2]{\stackrel{#1}{\Rightarrow}_{#2}}

\newcommand{\fase}{\texttt{FASE}}

\newcommand{\fifo}{\mbox{{\sf Fifo}}}
\newcommand{\pipe}{\mbox{{\sf Pipe}}}
\newcommand{\buff}{\mbox{{\sf Buff}}}
\newcommand{\mem}{\mbox{{\sf Mem}}}

\newcommand{\bc}{\mbox{{\sf BC}}}
\newcommand{\Pp}{\mbox{$\mathbb{P}$}}

\newcommand{\RT}{\mbox{\sf RT}}
\newcommand{\DL}{\mbox{\sf DL}}
\newcommand{\RTS}{\mbox{\sf RTS}}
\newcommand{\RRTS}{\mbox{\sf rRTS}}
\newcommand{\Chi}{\text{\large\raisebox{0.45ex}{$\chi$}}}
\title{Evaluating the Efficiency of Asynchronous Systems with $\fase$%
\thanks{
This work was supported by the PRIN Project
`Paco:Performability-Aware Computing: Logics, Models, and
Languages'.}}

\author{F. Buti, M. Callisto De Donato, F. Corradini, M.R. Di Berardini\\
{\small School of Science and Technology, University of Camerino}\\
{\small \{federico.buti, massimo.callisto, flavio.corradini,
mariarita.diberardini\}@unicam.it}\vspace{0.4cm}\\
W. Vogler\\
{\small Institut f\"ur Informatik, Universit\"at Augsburg}\\
{\small vogler@informatik.uni-Augsburg.de}
}

\date{}
\begin{document}
\maketitle
\begin{abstract}
In this paper, we present \fase\ ({\tt F}aster {\tt A}synchronous {\tt S}ystems
{\tt E}valuation), a
tool for evaluating the worst-case efficiency of asynchronous systems. The tool
is based on some
well-established results in the setting of a timed process algebra (PAFAS: a
Process Algebra for
Faster Asynchronous Systems). To show the applicability of \fase\ to concrete
meaningful examples,
we consider three implementations of a bounded buffer and use \fase\ to
automatically evaluate their
worst-case efficiency. We finally contrast our results with previous ones where
the efficiency of
the same implementations has already been considered.
\end{abstract}

\section{Introduction}\label{sec:intro}
PAFAS~\cite{CVJ02} has been proposed as a useful tool for comparing the
worst-case
efficiency of asynchronous systems. It is a CCS-like process description
language
\cite{Mil89} where basic actions are atomic and instantaneous but have
associated a time bound
interpreted as the maximal time delay for their execution. These upper time
bounds can be used to
evaluate efficiency, but they do not influence functionality (which actions are
performed); so
compared with CCS also PAFAS treats the full functionality of asynchronous
systems. In~\cite{CVJ02},
processes are compared via a variant of the testing approach developed by
De~Nicola and Hennessy
in~\cite{DNH84}. Tests considered in~\cite{CVJ02} are test environments (as
in~\cite{DNH84})
together with a time bound. A process is embedded into the environment (via
parallel composition)
and satisfies a (timed) test, if success is reached before the time bound in
{\em every} run of the
composed system, i.e.\ even in the worst case. This gives rise to a faster-than
preorder
over processes that is naturally an {\it efficiency preorder}. Moreover, this
efficiency preorder
can be characterised as inclusion of a special kind of {\em refusal traces},
which provide decidability of
the testing preorder for finite state processes.
%%%%%%%%%%

In~\cite{CV05}, it has been shown that the faster-than preorder provided
in~\cite{CVJ02} can
equivalently be defined on the basis of a performance function that gives the
worst-case time needed
to satisfy any test environment (or user behaviour). If the above timed testing
scenario is
adapted by considering only test environments that want $n$ tasks to be
performed as fast as possible
(possibly in parallel), this performance function is {\em asymptotically
linear}. This provides us
with a {\em quantitative} measure of system performance, essentially a function
from natural
numbers to natural numbers called {\em response performance function} that
measures how
fast the system under consideration responds to requests from the environment.

In this paper, we present \fase, a corresponding tool that supports the
evaluation of this
function for a given system. In order to show the applicability of \fase\ to
concrete meaningful examples, we consider three different implementations of a
bounded buffer and
use \fase\ to automatically evaluate their efficiency. The three implementations
are called \fifo,
\pipe\ and \buff. \fifo\ is a bounded-length first-in-first-out queue, \pipe\ is
a sequence of one
place buffers connected end-to-end and \buff\ is an array used in a circular
fashion. We prove that
\fifo\ is always more efficient than \pipe\ and \buff, and that \buff\ is more
efficient than \pipe\
only if the number of requests is sufficiently small w.r.t. the size of the
buffer. These results
are quite different from those presented in~\cite{CDV01} (see
Section~\ref{sec:concluding}) where
the efficiency of the same buffer implementations has been compared by means of
the efficiency
preorder defined in~\cite{CVJ02}. The reason is that here (as in~\cite{CV05}) we
only consider
a specific class of user behaviours.

The rest of this paper is organised as follows. Section~\ref{sec:pafas} recalls
PAFAS and
the technical details we need to define the response performance.
Section~\ref{sec:algos}
presents \fase\ and its main algorithms. Section~\ref{sec:casestudy} describes
the three buffer 
implementations and states
our main results. Finally, Section~\ref{sec:concluding} presents some concluding
remarks.

\section{PAFAS}\label{sec:pafas}
In this section we briefly introduce PAFAS, its operational
semantics and the performance function to evaluate worst-case efficiency. We
refer the reader
to~\cite{CVJ02} and~\cite{CV05} for more details.
We use the following notation: $\A$ is an infinite
set of basic actions with a special action $\omega$, which is  reserved for
observers (test
processes) in the testing scenario to signal the success of a test. The
additional action $\tau$
represents an internal activity that is unobservable from other components.
Actions in $\At=\A\cup\{\tau\}$(ranged over by $\alpha,\beta,\cdots$) can
let time $1$ pass before their execution, i.e.\ 1 is their maximal delay. After
that time, they
become {\em urgent} actions. The set of urgent actions is $\underline{\A}_\tau
=\{\underline{a}\ |\ a \in \A\} \cup {\tau}\}$ and it is ranged over by
$\underline{\alpha},\underline{\beta},\cdots$. Furthermore,
$\Chi$ is the set of process variables $x,y,z,\dots$ used for recursive
definitions. A {\it general
relabelling function} (incorporating relabelling and hiding)
is a function $\Phi:\At \to \At$ where the set $\{\alpha\in\At \,|\, \emptyset
\neq \Phi^{-1}(\alpha)\ne\{\alpha\}\}$ is finite and $\Phi(\tau)=\tau$.

\begin{definition}\rm\label{PCTIP} ({\it Timed Processes}) The set $\Pp$ of {\it
(timed) processes}
is the set of closed (i.e. without free variables) and guarded (i.e. variable
$x$ in a $\mu x.P$
only appears within the scope of a prefix $\alpha.()$, where $\alpha \in \At$)
terms generated by
the following grammar:
$$ P  ::= \Nilp \gor \gamma . P  \gor P+P
\gor P\|_A P \gor P[\Phi] \gor x  \gor \mu x.P $$
where $\gamma$ is $\alpha$ or $\underline{\alpha}$ for some $\alpha \in \At$,
$\Phi$ a general
relabelling function, $x\in\Chi$ and $A\subseteq\A$ possibly infinite.
\end{definition}

A brief description of our operators now follows. $\Nilp$ is the Nil-process,
which cannot perform
any action, but may let time pass without limit \footnote{A trailing $\Nilp$
will often be omitted,
so e.g.\ $a.b+c$ abbreviates $a.b.\Nilp+c.\Nilp$.}; $\alpha.P$ and
$\underline{\alpha}.P$ is
(action-) prefixing, known from CCS. In particular, process $a.P$ performs $a$
with a {\em maximal}
delay of 1; hence, it can either perform $a$ immediately, or can idle for time 1
and become
$\underline{a}.P$. In the latter case, the idle-time has elapsed and action $a$
must either occur
or be deactivated (in a choice-context) before time may pass further. Our
processes are {\em
patient}: as a stand-alone process, $\underline{a}.P$ has no reason to wait; 
but as a component in
$\underline{a}.P\|_{\{a\}} a.Q$, it has to wait for synchronisation on $a$ and
this can take up to
time $1$, since the component $a.Q$ may idle this long. $P_1 + P_2$ models the
choice between two
conflicting processes $P_1$ and $P_2$.  $P_1\|_A P_2$ is the TCSP-like parallel
composition
of two processes $P_1$ and $P_2$ that run in parallel and have to synchronise on
all
actions from $A$ \cite{bhr84}.
In the following we write $\|$
as a shorthand for $\|_{\A\backslash \{\omega\}}$. $P[\Phi]$ behaves as $P$ but
with the actions
changed according to $\Phi$. Finally, $\mu x.P$ models a recursive definition;
recursive equations are a
common way of defining processes.

We now define the refusal traces of a process $P$. Intuitively, a refusal trace
records, along a
computation, which actions $P$ can perform ($P \nar{\alpha}{r} P'$, $\alpha \in
\At$) and which
actions $P$ can refuse to perform ($P\nar{X}{r} P'$, $X \subseteq \A$). A
transition like
$P\nar{X}{r} P'$ is called a {\it (conditional) time step}. The actions in the
set $X$ are not
urgent (see rule Pref$_{r2}$ in Fig.\ \ref{refusal}) so $P$ is justified in not
performing them but
performing a time step instead. Since other actions might be urgent and cannot
be refused,
$P$ as a stand-alone-process might actually be unable to let time pass.
% in any contest (i.e. it cannot refuse to perform any action). ???
But if $P$ is a component of a larger system, these actions might be further
delayed due to
synchronisation with some other components, and a time step is possible.
Whenever $P$ can make a time step in {\em any} context
(i.e\ if $P\nar{X}{r} P'$ and  $X =\A$), we say that $P$ performs a {\em full
time step} and also
write $P \nar{1}{} P'$.

\begin{definition}\rm\label{PREFSOS} ({\it Refusal operational semantics})
The SOS-rules in Fig.\ \ref{refusal} (plus symmetric rules for Par$_{a1}$ and
Sum$_a$
for actions of $P_2$) define the transition relations $\nar{\alpha}{r}\subseteq
(\Pp\times \Pp)$ for
$\alpha \in \At$ and $\nar{X}{r}\subseteq (\Pp \times\Pp)$ for $X\subseteq\A$.
\end{definition}

\begin{figure*}[thb]
$$\begin{array}{c}
\name{Pref}_{a1}\;\sos{}{\alpha .P \nar{\alpha}{r}P}
\quad\quad
\name{Pref}_{a2}\;\sos{}{\underline{\alpha} .P \nar{\alpha}{r}P}
\quad\quad
\name{Sum}_a\;\sos{P_1\nar{\alpha}{r}P'_1}{P_1+P_2\nar{\alpha}{r}P'_1}
\\[4ex]
\name{Par}_{a1}\;\sos{\alpha\notin A,\;P_1\nar{\alpha}{r}P'_1}{P_1\|_A
P_2 \nar{\alpha}{r}P'_1\|_A P_2}
\quad\quad
\name{Par}_{a2}\;\sos{\alpha\in
A,\;P_1\nar{\alpha}{r}P'_1,\;P_2\nar{\alpha}{r}P'_2} {P_1\|_A P_2
\nar{\alpha}{} P'_1\|_A P'_2}
\quad
\\[4ex]
\name{Rel}_a\;\sos{P\nar{\alpha}{r}P'}{P[\Phi]\nar{\Phi(\alpha)}{r}P'[
\Phi]}
\quad\quad
\name{Rec}_a\;\sos{P\{\mu x.P/x\} \nar{\alpha}{r} P'}{\mu x.P
\nar{\alpha}{r}P'}\\
[8ex]
\name{Nil}_r\;\sos{}{\Nilp\nar{X}{r}\Nilp}
\quad\quad
\name{Pref}_{r1}\;\sos{}{\alpha.P\nar{X}{r}\underline{\alpha}.P}
\quad\quad
\name{Pref}_{r2}\;\sos{\alpha\notin
X\cup\{\tau\}}{\underline{\alpha}.P\nar{X}{r}\underline{\alpha}.P}
\\
[4ex]
\name{Par}_r\;\sos {\forall_{i=1,2}\;P_i\nar{X_i}{r}P'_i,\; \textstyle
X\subseteq(A\cap\bigcup\nolimits_{i=1,2}X_i)\cup
((\bigcap\nolimits_{i=1,2}X_i)\setminus A)} {P_1\|_A P_2
\nar{X}{r}P'_1\|_A P'_2}
\\
[4ex]
\name{Sum}_r\;\sos {\forall_{i=1,2}\;P_i\nar{X}{r}P'_i}
{P_1+P_2\nar{X}{r}P'_1+P'_2}
\quad\quad
\name{Rel}_r\;\sos
{P\nar{\Phi^{-1}(X\cup\{\tau\})\setminus\{\tau\}}{r}P'}{P[\Phi]\nar{X}
{r}P'[\Phi]}
\\
[4ex]
\name{Rec}_r\;
\sos{P\{\mu x.P/x\} \nar{X}{r}P'}{\mu x.P\nar{X}{r}P'\{\mu x.P/x\}}\\
\end{array}$$
\caption{The Refusal Operational Semantics of PAFAS
processes.}\label{refusal}
\end{figure*}
The rules in Fig.~\ref{refusal} explain the operational semantics of PAFAS
processes. A process like
$\alpha.P$ can either perform action $\alpha$ immediately and then become $P$
(rule
\name{Pref}$_{a_1}$), or can let time $1$ pass and refuse any set of actions
(rule
\name{Pref}$_{r1}$).
A process $\underline{\alpha}.P$ can perform an action $\alpha$ (rule
\name{Pref}$_{a2}$) and on its own cannot delay such an execution (rule
\name{Pref}$_{r2}$).
Since internal action $\tau$ has never to be synchronised, a process prefixed by
an urgent $\tau$
cannot make a time step. Another rule worth noting is \name{Par}$_r$ that
defines which actions a
parallel composition can refuse during a time step. The intuition is that
$P_1\|_A P_2$ can refuse
an action $a$ if either $a \not \in A$ ($P_1$, $P_2$ are not forced to
synchronise on $a$) and both
$P_1$, $P_2$ can refuse $a$, or $a \in A$ ($P_1$, $P_2$ are forced to
synchronise on $a$) and either $P_1$ or $P_2$ can refuse $a$. The other rules
are as expected.

For sequences $w \in (\At \cup 2^{\A})^\ast$, we define $P \nar{w}{r} P'$ as
expected:
$P\nar{w}{r}P'$ if either $w=\varepsilon$ (the empty sequence) and $P'=P$ or
there is $Q \in \Pp$
and $\mu \in (\At \cup 2^{\A})$ such that $P\nar{\mu}{r} Q \nar{w'}{r} P'$ and
$w=\mu w'$.
Similarly, we define $P \nar{w}{} P'$ for $w \in (\At \cup \{1\})^\ast$. In the
latter case,
$\zeta(w)$ is the duration of $w$, i.e. the number of full time steps in $w$. We
write $P
\Nar{v}{r}P'$ ($P \Nar{v}{}P'$) if $P\nar{w}{r}P'$ ($P\nar{w}{}P'$, resp.) and
$v=w/\tau$ ($v$ is
the sequence $w$ with all $\tau$'s removed). Finally,
$\RT(P)=\{w\,|\,P\Nar{w}{r}\}$  and
$\DL(P)=\{w\,|\,P\Nar{w}{}\}$ are the sets of {\it refusal traces} and {\em
discrete traces}
(resp.) of $P$.

For processes $P\ ,Q \in \Pp$, $\RT(P) \subseteq \RT(Q)$ implies $\DL(P)
\subseteq \DL(Q)$:
 $\DL(P)$ corresponds to the set of traces $w \in \RT(P)$ where $X = \A$ for all
refusal sets $X$ 
in $w$. Finally, the {\it refusal transition system} $\RTS(P)$ of $P$ is defined
as the set of 
all transitions $Q \nar{\mu}{r} Q'$ with $\mu \in \At$ or $\mu \subseteq \A$
where $Q$ is reachable 
from $P$ via such transitions. It is easy to prove that $\RTS(P \parallel_{A}
Q)$ can be determined 
from $\RTS(P)$ and $\RTS(Q)$ according to the SOS-rules for parallel composition
given 
in Fig. \ref{refusal}.

%%%%%%%%%%%%%%%%%%%%%%%%%%%%%%%%%%%%%%%%%%%

In the timed testing of~\cite{CVJ02}, $P$ satisfies a timed test (observer $O$
with special success
action $\omega$ plus time bound $D$) if every discrete trace of $P\| O$
performs $\omega$ before time $D$; $P$ is {\em faster than} $Q$, $P\sqsupseteq
Q$,
if $P$ satisfies all timed tests that $Q$ satisfies. This
preorder is a qualitative notion since a timed test is either satisfied or not,
and a process is
more efficient than another or not.

One of the main results in \cite{CVJ02} is that the faster-than preorder can be
characterised
by refusal-trace-inclusion, i.e.\  $P\sqsupseteq P'$ iff $\RT(P) \subseteq
\RT(P')$ (see Theorem
5.13 in~\cite{CVJ02}). A new formulation of this preorder has been provided in
~\cite{CV05} 
(see Prp. 9) that brings to light its quantitative nature; the new formulation
is given using the 
following performance function:

In~\cite{CV05}, Prop. 9 provides 

\begin{definition}\rm\label{def:performance} ({\it Performance})
Let $P\in \Pp$ be a process and $O \in \Pp$ be a test process. We define the
{\em performance
function} $p$ as:

\begin{center}
$p(P,O) = \sup \{\, n \in \nat_0 \,|\, \exists v \in \DL(P \| O): \zeta(v) = n
\mbox{ and } v 
\mbox{ does not contain } \omega \,\}$\\
\end{center}

\noindent If the right-hand side has no maximum, the
supremum is $\infty$. The performance function $p_P$ is defined by $p_P(O) =
p(P,O)$, and we write
$P\sqsupseteq Q$ if $p_P(O) \leq p_Q(O)$ for each $O$.
\end{definition}

The performance function $p$ (as well as the preorder $\sqsupseteq$) contrasts
processes w.r.t.\ all
possible test environments. In some cases, this might be too demanding and one
can make some
reasonable assumption about the user behaviour. Consider a scenario where users
have a
number of requests (made via $in$-actions) that they want to be answered (via
$out$-actions) as
fast as possible. This class of users is defined as ${\cal U} = \{U_n \,|\, n
\geq 1\}$ where $U_1
\equiv
\underline{in}.\underline{out}.\underline{\omega}$ and $U_{n} = U_{n-1} \;
\|_{\{\omega\}}\;
\underline{in}.\underline{out}.\underline{\omega}$ (for any $n>1$).
% % Comparing processes w.r.t. these users means to compare their
% % performance under heavy load; this is clearly of practical
% importance.
Given these users, we can define the {\em response performance} $rp$ of a
testable process $P$ as a
function from $\nat$ to $\nat_0$ with $rp_P(n) = p_P(U_n) = p(P,U_n)$; here $n$
is the size (i.e.
the number of requests) of the user.

In what follows we briefly describe how the response performance of a process
$P$ can be calculated 
from its refusal transition system. 
We restrict attention to so-called response processes, which never produce an
$out$ without a 
corresponding preceding $in$.

By Definition~\ref{def:performance}, to determine $rp_P(n)$
we have to consider all $w \in \DL(P \;\|\; U_u)$ that do not contain $\omega$,
count the number of
their full time steps and then take the supremum of the numbers so obtained.
These traces are just
paths in $\RTS(P \;\|\; U_u)$ that do not contain $\omega$ and contain only full
time steps. These
paths can have at most $n$ $in$'s and $n$ $out$'s (due to synchronisation with
$U_n$). But after
the $n$-th $out$, an urgent $\omega$ becomes available and no more full time
steps can occur before
$\omega$; in other words, full time steps are only possible before the $n$-th
$out$. So we have 
solely to consider paths in $\RTS(P \;\|\; U_n)$ that contain only full time
steps and have at most $n$
$in$'s and $(n-1)$ $out$'s (and, hence, no $\omega$).
In~\cite{CV05} it has been proven that for each of these paths there is a
so-called {\rm
$n$-critical path} in $\RRTS(P)$\footnote{This is a reduced version of the
$\RTS(P)$
where all conditional time steps, that cannot participate in a full time step
when $P$ runs in
parallel with a user $U_n$, are removed. For more details see \cite{CV05}.} with
the same number of
time steps. Thus, the following characterisation for the response performance
can be given.

\begin{theorem}\rm\label{theo:rp-characterization}
({\it Characterisation for response performance}) A path in $\RRTS(P)$ is
$n$-critical if it
contains at most $n$ $in$'s, at most $n-1$ $out$'s , and all time steps before
the $n$-th $in$ are
full. The response performance of a process $P$ is the supremum of the numbers
of time steps taken
over all $n$-critical paths.
\end{theorem}

Now a key observation is that, when the number $n$ of requests is large compared
to the
number of processes in $\RRTS(P)$, an $n$-critical path with many time steps
must contain cycles.
Finding the worst cycles turns out to be essential for performance evaluation.
In~\cite{CV05},
these worst cycles are distinguished to be either {\em catastrophic} or {\em
bad} cycles.

\begin{definition} \rm\label{def:catastrophic} ({\it Catastrophic cycle})
A cycle in $\RRTS(P)$ is a catastrophic cycle if it has a positive number of
time steps but no
$in$'s and no $out$'s. If $\RRTS(P)$ has a catastrophic cycle then $rp_P(n) =
\infty$ for some
$n$.
\end{definition}

Intuitively, once in a catastrophic cycle, we cannot satisfy any other request
(this is because a
catastrophic cycle does not contain $out$-actions) but time can pass
indefinitely (the cycle has at
least one time step). As a consequence, there exists some $n$ (depending on how
many $in$ and
$out$-actions are performed on a path in $\RRTS(P)$ from $P$ to this cycle) such
that $rp_P(n) =
\infty$, i.e.\ some user is not satisfied within a bounded time.
If $\RRTS(P)$ is free from catastrophic cycles we search for the so called bad
cycles:

\begin{definition} \rm\label{def:bad-cycles} ({\it Bad cycle})
For $P$ without catastrophic cycles, we consider cycles reached from $P$ by a
path where all time
steps are full and which themselves contain only full time steps. Let the {\em
average
performance} of such a cycle be the ratio between the number of its full time
steps and the number
of its $in$ actions. A {\em bad cycle} is a cycle in $\RRTS(P)$ which has
maximal average
performance.
\end{definition}

Theorem 16 in~\cite{CV05} shows that $rp_P$ is asymptotically linear, i.e.
$\exists\ a \in \Real$ s.t. $rp_P(n)
= an +  \Theta(1)$, and that the ``asymptotic performance'' $a$ of $P$ is the
 average performance of a bad cycle. In other words, while $n$-critical paths
give the exact value 
of the response performance of a process, the average performance of a bad cycle
is its asymptotic 
behaviour. Both catastrophic and bad cycles can be automatically checked with
\fase.%

\section{Performance evaluation with FASE}\rm\label{sec:algos}
In this section we introduce
$\fase$\footnote{\url{http://cosy.cs.unicam.it/fase/}}, the tool that has been
used to automatically evaluate the worst-case efficiency of the three
buffer implementations discussed in Section~\ref{sec:casestudy}.
\fase\ is written in Java language and consists of two main components.
The former one is essentially a parser unit; it takes as input a
sequence of characters that represents a PAFAS process $P$ and
builds its $\RTS(P)$. The second one is the performance module that implements
the algorithms used to evaluate the worst-case efficiency of $P$. 
The two modules are loosely coupled; they communicate via a shared Java data
structure or via an XML-based representation of the $\RTS$. The last aspect is
very important since changes to a module do not affect the other one; moreover,
the XML interface guarantees a broader interoperability with external tools such
as graph visualisers, which could be useful for further analysis of the modelled
systems.

\subsubsection{Parsing unit}
Fig. \ref{FASEARC} shows on top the parsing phase that is based on two
well-known tools: JFlex \cite{Klein01} as the lexer generator and jacc
\cite{Jones04jacc} as the parser generator. 
JFlex defines how input streams must be arranged into words - called
\textit{tokens} - while jacc pseudocode gives rules - called {\it productions} -
to compound such tokens. These productions are used by the parser to generate
the data structure that contains the hierarchical representation of the process
where each element is a term of the grammar in Definition \ref{PCTIP}. For
example, after parsing $P = a.nil\ +\ \underline{b}.nil$, the hierarchy
structure obtained has on top the process variable $P$ which contains a choice
operator with a prefix $a.nil$ and an urgent prefix $\underline{b}.nil$
respectively, and so on. 
Every element is an instance of a Java class that handles the respective SOS
rules given in Fig. \ref{refusal}; thus, an element encapsulates both functional
and temporal behaviour used to generate  $\RTS(P)$ as indicated at the bottom of
Fig. \ref{FASEARC}.

The building process of $\RTS(P)$ exploits the hierarchical structure,
traversing it from the root element; at each step the operator objects generate
the proper nodes and transitions according to Definition \ref{PREFSOS}. For
instance, $P = a.nil + \underline{b}.nil$ will produce the node $P$ with two
outgoing transition $a$ and $b$ to the same node $nil$; the additional refusal
transition $\{a\}$ to the process $\underline{a}.nil\ +\ \underline{b}.nil$ will
be produced according to rules $\name{Sum}_{r}$, $\name{Pref}_{r1}$ and
$\name{Pref}_{r2}$. The same method will be applied to the remaining nodes as
expected.

Such an architectural structure provides several advantages. The pseudocode of
both lexer and parser are based on common syntaxes (such as regular expressions
and BNF rules) that are extremely smaller than actual Java code, easier to
understand and easier to maintain. Semantics of each operator is coded in a
separate compile unit, hence it can be specified independently and modified in a
second stage, if necessary. 

\begin{figure}
	\centering
	\scriptsize
	\begin{tikzpicture}[scale=1.2, node distance=2cm,auto,swap]
		\begin{scope}[>=latex]
			\draw [rectangle, text centered, rounded corners, top
color = blue!20, shading angle=315, draw=black!80] (0,0) rectangle (2.5,1.2);
			\draw [rectangle, text centered, rounded corners, top
color = violet!50, shading angle=315, draw=black!80] (4.5,-0.1) rectangle
(7,1.3);
			\draw [rectangle, text centered, rounded corners, top
color = pink!80, shading angle=315, draw=black!80] (1.3,0.5) rectangle
(2.3,0.2);
			\draw [rectangle, text centered, rounded corners, top
color = pink!80, shading angle=315, draw=black!80] (1.35,0.45) rectangle
(2.35,0.15);
			\draw [rectangle, text centered, rounded corners, top
color = pink!80, shading angle=315, draw=black!80] (1.4,0.4) rectangle
(2.4,0.1);
			\draw [rectangle, text centered, rounded corners, top
color = green!50, shading angle=315, draw=black!80] (4.7,0.2) rectangle (5.8,1)
;
			\draw [rectangle, text centered, rounded corners, top
color = yellow!50, shading angle=315, draw=black!80] (2.3,-1.75) rectangle
(4.65,-3);
			\draw [rectangle, text centered, rounded corners, top
color = pink!50, shading angle=315, draw=black!80] (2.3,-4) rectangle
(4.65,-5.25);
			%labels...
			\draw[-] [black!80] (0.5,0.9) node {Lexer};  
			\draw[-] [black!80] (6.4,1) node {Parser};  
			\draw[-] [black!80] (1.9,0.265) node {\tiny token}; 
			\draw[-] [black!80] (5.25,0.65) node {Rules};  
			\draw[-] [black!80] (3.5,-2.2) node {usable Java};  
			\draw[-] [black!80] (3.5,-2.5) node {source code};  
			\draw[-] [black!80] (3.5,-4.45) node {runtime};  
			\draw[-] [black!80] (3.5,-4.75) node {code}; 
			\draw[-] [black!80] (1.25,-0.35) node {JFlex}; 
			\draw[-] [black!80] (5.75,-0.35) node {jacc};
			\draw[-] [black!80] (-0.15,-4.5) node {Process $P$};
			\draw[-] [black!80] (6.9,-4.5) node {$\RTS(P)$};
			%arrows...
			\draw[<->] [black!80] (2.5,0.6) -- node[above=0.3pt]
{\tiny get\_next\_token()} (4.5,0.6);
			\draw[-] [black!80] (1.25,-0.5) -- (3.475,-1);
			\draw[-] [black!80] (5.75,-0.5) -- (3.475,-1);
			\draw[->] [black!80] (3.475,-3) -> node[right=0.5pt]
{compilation} (3.475,-4);
			\draw[->] [black!80] (3.475,-1) -> node[right=0.5pt]
{internal conversion} (3.475,-1.75);
			%special arrows...
			\path[edge] [bigBlueEdge] (0.5,-4.4) .. controls
+(down:0.2cm) and +(left:0.5cm) .. (2,-4.6);  
			\path[edge] [bigBlueEdge] (4.9,-4.7) .. controls
+(up:0.2cm) and +(left:0.5cm) .. (6.4,-4.5); 
			%legend...
			\draw[-] [black!80] (-1,1.25) -- node[above=0.5pt,
rotate=90] {\tiny pseudolanguage} (-1,-0.75);
			\draw[-] [black!80] (-1,-1.75) -- node[above=0.5pt,
rotate=90] {\tiny Java language} (-1,-5.3);
		\end{scope}[>=latex]
	\end{tikzpicture}
	\caption{An architectural overview of $\fase$.}\label{FASEARC}
\end{figure}
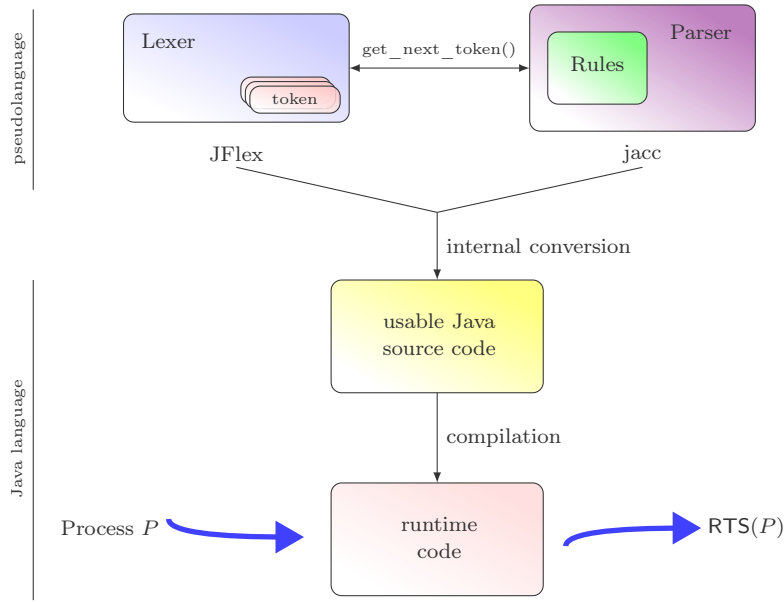

\subsubsection{Performance unit}
The performance component provides all the algorithms needed to evaluate systems
performance according to the theoretical results stated in the previous section.
In particular, \fase\ adopts two new algorithms for catastrophic cycles
detection and bad cycle calculation that improve those proposed in \cite{CV05}.
Moreover, \fase\ is also able to generate the complete set of traces that
characterises the behaviour of the process. Such diagnostic information is
useful to the user since it helps to understand why a modelled system produces
catastrophic cycles or has certain worst-case performance. This feature has
helped us to validate the results on the response performance of the three
buffer implementations discussed in the next section. 

\paragraph{Catastrophic cycles} The problem of finding catastrophic cycles in a
process
$P$ has been solved in~\cite{CV05} in time $\theta(N^{3})$ where $N$ is the
number of
nodes in $\RRTS(P)$. The new algorithm adopted in \fase\ takes advantage of the
well-known problem of finding the {\em Strongly Connected Components} ({\it
SCCs})~\cite{Aho83}.
Since an SCC of a graph is a subgraph that is strongly connected and maximal,
the following suffices. We obtain a new graph $G$ from $\RRTS(P)$ by deleting
all edges labelled $in$ and $out$ and apply the algorithm for finding the SCCs.
If at least one contains some time step, we can conclude that $P$ has a
catastrophic cycle. Indeed, if $S$ is an SCC in $G$ and there is a time step
$(u, v)$ with $u$ and $v$ nodes of $S$, then $S$ has a path from $v$ to $u$,
i.e. it has a catastrophic cycle that is also contained in $\RRTS(P)$. Vice
versa, if $P$ has a catastrophic cycle, it is contained in some SCC of $G$,
which therefore contains a time step.

The standard SCC discovery algorithm has complexity $O(N + E)$ with $N$ and $E$
the number of nodes and edges of $G$ respectively, and the same applies to
construction of $G$ and thus to finding catastrophic cycles in \fase.
Table \ref{tab:cctime}\footnote{$\pipe$ and $\buff$ are two different
implementation of the same buffer discussed in the next section. We have left
out $\fifo$ since its representation is too small for sensible comparison.}
reports the running time for the original and the new algorithm.

\begin{table}
	\centering
	\begin{tabular}{l*{5}{l}}
	$Cells\;$ & $nodes/edges\;$ & $previous\ and\;\;\;$ & $new\ algorithm\;$
& $Gain$\\
	\midrule
	$\pipe_{5}$ & $114/292$ 		& $37$ 			& $16$ 	
	& $74.1 \%$ 	\\
	$\buff_{5}$ & $96/216$ 			& $15$ 			& $15$ 	
	& $-$ 			\\
	$\pipe_{6}$ & $272/759$ 		& $578$ 		& $63$ 	
	& $89.1\ \%$ 	\\
	$\buff_{6}$ & $160/368$ 		& $93$ 			& $22$ 	
	& $76.3\ \%$ 			\\
	$\pipe_{7}$ & $648/1958$ 		& $11484$ 	& $296$ 	
& $ 97.4\ \%$ 	\\
	$\buff_{7}$ & $240/560$ 		& $390$ 		& $47$ 	
	& $ 87.9\ \%$  \\
	
	$\pipe_{8}$ & $1544/5034$ 	& $620109$ 	& $1575$	
& $99.7 \%$ 	\\
	$\buff_{8}$ & $336/792$ 		& $1172$ 		& $70$ 	
	& $94.0\ \%$ 			\\

	$\pipe_{9}$ & $3680/12902$ 	& $-$ 			& $9687$ 	
& $100\ \%$ 		\\
	$\buff_{9}$ & $448/1064$ 		& $2922$ 		& $109$ 
	& $96.2 \%$ 	\\
	\midrule
	\end{tabular}
	\caption{Catastrophic-cycle  detection time (expressed in {\it ms})}
	\label{tab:cctime}
\end{table}

\paragraph{Bad cycles}
Next we look for a bad cycle, possibly not unique, of $\RRTS(P)$ according to
Definition \ref{def:bad-cycles} that gives the average performance of $P$.
To determine this value, a graph $G$ is obtained from $\RRTS(P)$ by deleting all
non-full time steps and all nodes not reachable any more (see Proof of Theorem
17 of \cite{CV05} for more details). To apply a known algorithm from the
literature, we do not look for a cycle with maximal average performance in $G$,
but for one with minimal average throughput, the latter being just the inverse
of the average performance. Such a cycle can be seen as a set of sub-paths where
each one ends in a time step.

For the known algorithm, we must transform $G$ to a graph $G'$ where each edge
is weighted with some cost and represents one time step, i.e.\ an edge
corresponds to such a sub-path. Since the costs should be minimal, the subpath
without the last node must be a shortest path between the respective nodes as
measured by the number of $in$'s.
Hence, one obtains a new graph $G_0$ by deleting all time steps in $G$ and
computes its all-pairs shortest paths matrix $d$ with the Floyd-Warshall
algorithm, considering a weight $1$ for  $in$-transitions and $0$ for all the
other edges. The final $G'$ graph is constructed from the nodes of $G$ on the
basis of the matrix $d$; for every two nodes $u$, $v$ of $G$, where $d(u, v)$ is
finite and there exists a time step from $v$ to $v'$, we add the edge $(u, v')$
with cost $d(u, v)$. This construction can be carried out in time $O(N^3)$. Now
the problem of finding the minimal average throughput $t$ can be solved with
Karp's algorithm \cite{Karp78} applied to graph $G'$.

Although the above method is bounded by a complexity of $O(N^3)$, the
construction of the shortest-paths matrix $d$ has a cost of $\Theta(N^3)$. In a
common scenario where the behaviour of $P$ can be very complex, the computation
of the matrix could be expensive as reported in Table \ref{tab:bctime}. To get
around the problem, we have developed an improved algorithm. Starting from $G$
and $G_0$ as defined above, we reverse the edges of $G_0$ to obtain the graph
$G_{0}^{T}$. Since we are only interested in paths leading to a time step, for
each full time step $(v, v')$ of $G$, we apply Dijkstra's algorithm to
$G_{0}^{T}$ with root
node $v$ and weight $1$ for $in$-transitions, $0$ otherwise as above. Finally,
for each node $u$, such that there exists a path from $v$, we add an edge $(u,
v')$ in $G'$ where the cost is the length of a shortest path from $v$ to $u$.

With this approach, we calculate only those (shortest) paths that lead to time
steps, i.e. only those paths that correspond to edges in $G'$. On the contrary,
in the original algorithm, (shortest) paths between all pairs of nodes are
computed. Since Dijkstra's algorithm runs in time $O(E + N log N) $~\cite{Aho83}
(with $N$ and $E$ the number of nodes and edges respectively), constructing $G'$
takes $O(N (E + N log N))$, but at least the first factor $N$ will be
considerably smaller in practice.
Table \ref{tab:bctime} shows the improvements obtained when considering large
buffer implementations.

\begin{table}
	\centering
	\begin{tabular}{l*{6}{l}}
	$Cells$ & $nodes/edges$ of $G\;$ & $previous\ and\;\;$ & $new\
algorithm$ & $nodes/edges$ of $G'\;$ & $Gain$\\
	\midrule
	$\pipe_{5}$ 	& $114/292$ 			& $546$ 		
	& $62$				& $114/3648$ 			 &
$88.6\ \%$ 			\\
	$\buff_{5}$ 	& $96/216$ 				& $469$ 	
		& $62$ 				& $96/4608$ 			
& $86.7\ \%$ 		\\

	$\pipe_{6}$ 	& $272/759$ 			& $4279$ 		
	& $266$ 			& $272/17408$ 		 & $93.7\ \%$ 	
	\\
	$\buff_{6}$ 	& $160/368$ 			& $1422$ 		
	& $172$ 			& $160/12800$ 		 & $87.9\ \%$ 	
\\

	$\pipe_{7}$ 	& $648/1958$ 			& $15000$ 		
& $1438$ 			& $648/82944$ 		 & $90.4\ \%$ 	
\\
	$\buff_{7}$ 	& $240/560$ 			& $7485$ 		
	& $437$ 			& $240/28800$ 		 & $94.1\ \%$ 	
\\

	$\pipe_{8}$ 	& $1544/5034$			& $-$ 			
	& $6672$	 		& $1544/395264$ 	& $100\ \%$	
\\
	$\buff_{8}$ 	& $336/792$ 			& $12454$ 		
& $734$ 			& $336/56448$ 		 & $94.1\ \%$ 	
\\

	$\pipe_{9}$		& $3680/12648$ 		& $-$			
		& $56000$ 		& $3680/1884160$ 	 & $100\ \%$ 	
\\
	$\buff_{9}$ 	& $448/1064$ 			& $45031$ 		
& $1766$ 		& $448/100352$ 		& $96.0\ \%$ 		
\\
	\midrule
	\end{tabular}
	\caption{Construction time of $G'$ (expressed in {\it ms})}
	\label{tab:bctime}
\end{table}

\section{Evaluating the performance of three bounded buffer
implementations}\label{sec:casestudy}
In this section, we evaluate the worst-case efficiency of three implementations
of a bounded buffer
(of capacity $N+2$, where $N\ge 1$ is a fixed natural number) with \fase. These
implementations have
already been consider in~\cite{CDV01} where their efficiency has been compared
via the faster-than
preorder relation $\sqsupseteq$ defined in~\cite{CVJ02}. In particular, we want
to investigate if
the results stated in~\cite{CDV01} still hold in our quantitative setting with
the 
restricted class of users. The three implementations
are $\fifo$ (a bounded-length first-in-first-out queue), $\pipe$ (a sequence of
one place buffers
connected end to end) and $\buff$ (an array used in a circular fashion).
Unlike~\cite{CDV01}, we
abstract away from the actual values stored in the buffers and assume that the
latter perform, 
as visible actions, either an $in$-action (meaning that a value is saved in a
free cell of 
the buffer) or an $out$-action (meaning that the buffer gives back a value to
the external 
environment). This choice surely does not influence performance as already
discussed 
in~\cite{CV05}, since the operations are data-independent, and it allows us to
reduce
considerably the number of states considered when calculating the response
performance.

The first buffer implementation $\fifo$ shown in Fig. \ref{fif} directly
implements a first-in-first-out queue of capacity
$N+2$. It has no overhead in terms of internal actions and it is purely
sequential. In the
examples, we use names and defining equations (using $\equiv$) to describe
recursive
behaviour.

\begin{definition}[{\it The buffer \fifo}]\label{def:fifo}
We define $\fifo \equiv \fifo(0)$ where, for each $i = 0, \cdots N+2$,
$\fifo(i)$ is defined as follows:

\begin{enumerate}
\item $\fifo(0) \equiv in.\fifo(1)$
\item for $0<i<N+2$ then $\fifo(i) \equiv in.\fifo(i+1) + out.\fifo(i - 1)$
\item $\fifo(N+2) \equiv out.\fifo(N+1)$
\end{enumerate}
\end{definition}

\begin{proposition}\label{prop:fifo}
The asymptotic performance of \fifo\ is $2$ (i.e. $rp_{\fifo}(n) = 2
n + \Theta(1)$). Moreover, for any $N\geq 1$, $rp_{\fifo}(n) = 2n$.
\end{proposition}

\begin{proof}
We have used \fase\ in order to automatically prove that \fifo\ does
not have catastrophic cycles and to calculate its asymptotic
performance. For what concerns its response performance, we can easily see that
\fifo\ may need a
time step for any input or output. E.g.\ $(\A \, in \, \A \, out)^{n-1} \, \A \,
in \, \A$ is an $n$-critical
path with a maximum number of time steps. We can conclude that
$rp_{\fifo}(n)=2n$.
\end{proof}

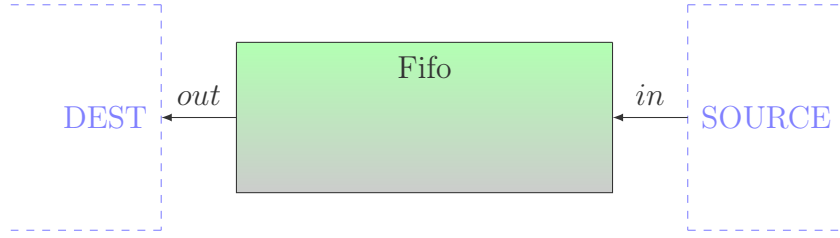
\begin{figure}
	\centering
	\begin{tikzpicture}[only marks, scale=1]
	\begin{scope}[>=latex]
	   %draw dotted... 
	   \draw[-] [blue!50, dashed] (0,0) -- (2,0);
	   \draw[-] [blue!50, dashed] (2,0) -- node[left=1pt] {DEST} (2,-3);
	   \draw[-] [blue!50, dashed] (2,-3) -- (0,-3);
	   \draw[-] [top color = green!30, bottom color = black!20,
draw=black!80](8,-0.5) -- (8,-2.5) -- (3,-2.5) -- (3,-0.5);	
		 \draw[-] [black!80] (3,-0.5) -- node[below=1pt] {Fifo}
(8,-0.5);
	   \draw[-] [blue!50, dashed] (11,0) -- (9,0);
	   \draw[-] [blue!50, dashed] (9,0) -- node[right=1pt] {SOURCE} (9,-3);
	   \draw[-] [blue!50, dashed] (9,-3) -- (11,-3);
	   %draw arrows...
	   \draw[->] [black!80] (3,-1.5) -> node[above=0.5pt] {$out$} (2,-1.5);
	   \draw[->] [black!80] (9,-1.5) -> node[above=0.5pt] {$in$} (8,-1.5);
	\end{scope}[>=latex]   
	\end{tikzpicture}
	\caption{The Software Architecture for $\fifo$}\label{fif}
\end{figure}

A buffer can also be implemented as a concatenation of $N+2$ cells as shown in
Fig. \ref{pip}, 
where a cell is an input/output device
that contains at most one value. In such a case, the cells have to be connected 
end-to-end, so that the output of each cell becomes the input of the next one.

\begin{definition}[{\it The buffer \pipe}]\label{def:pipe}
We define an {\em empty cell} as the process $C \equiv in. C'$ where $C'\equiv
out.C$. Let $i=0,
\cdots, N+1$; the $i$-th cell of \pipe\ is defined by $C_i \equiv C[\Phi_i]$
where the relabelling
function $\Phi_i$ is defined as follows:
%%%%%%%%%%%%%%%%%%%%%%%%%%%%%%%%%%%%%%%%%%%%
$$
\Phi _{i}(\alpha )=
\begin{cases}
\begin{array}{ll}
\delta_{i}   & \mbox{if } \alpha = in \mbox{ and } 0 \leq i \leq N\\
\delta_{i-1} & \mbox{if } \alpha= out \mbox{ and } 1 \leq i \leq N+1\\
\alpha  & \mbox{otherwise}\\
\end{array}
\end{cases}
$$
%%%%%%%%%%%%%%%%%%%%%%%%%%%%%%%%%%%%%%%%%%%%
Here each action $\delta_j$ passes the value from the $(j + 1)$-th to the $j$-th
cell. We force synchronisation among two consecutive cells by properly
relabelling $in$ and
$out$-actions of single cells. Let $A = \{\delta_0, \delta_1, \cdots,
\delta_{N+1}\}$. We define
$\pipe \equiv (C_0 \parallel_{\delta_0} C_1 \parallel_{\delta_1} \ldots
\parallel_{\delta_{N+1}}
C_{N+1} )/A$ where, for any given $P \in \Pp$, the process $P/A$ is the same as
$P [\Phi_{A}]$ where
%the relabelling function $\Phi_{A}$ is such that
$\Phi_{A}(\alpha) = \tau$ if $\alpha \in A$ and $\Phi_{A}(\alpha )=\alpha$ if
$\alpha \notin A$.
\end{definition}

Besides input and output of values, \pipe\ performs a number of activities in
order to manage the
queue of cells, i.e. to move values from  a cell to the next one. These actions
are
synchronisations between consecutive cells on actions $\delta_i$ and have been
made internal.
Moreover, note that $\pipe$ receives input values in cell $N+1$ (the only
$in$-action not renamed by
functions $\Phi_i$ is the one performed by this cell) and delivers output values
at cell 0.

\begin{proposition}\label{prop:pipe}
The asymptotic performance of \pipe\ is $2$. Moreover, for any $N \geq 1$, we
have that
$rp_{\pipe}(n) = 2 n + (N + 1)$.
\end{proposition}

\noindent {\em Sketch of the proof}:
Again, we have used \fase\ to prove that $\RRTS(\pipe)$ does not contain
catastrophic cycles and to
evaluate \pipe's asymptotic performance.
We only indicate why $rp_{\pipe}(n) = 2 n + (N + 1)$. The first value is moved
to cell $N+1$ after
one time step; with every further time step, it is moved to the next cell; so it
arrives in cell 0
after $N+2$ time steps and is delivered with the next one. After the second time
step, cell $N+1$
becomes empty, so the second value is put into cell $N+1$ after three time steps
and
then moves along the pipe with the same speed as the first one. Thus, the next
value is always delivered after two more time steps; see \cite{COVO07} for the
formal treatment of a more general case.\hfill $\Box$ \vskip.2cm

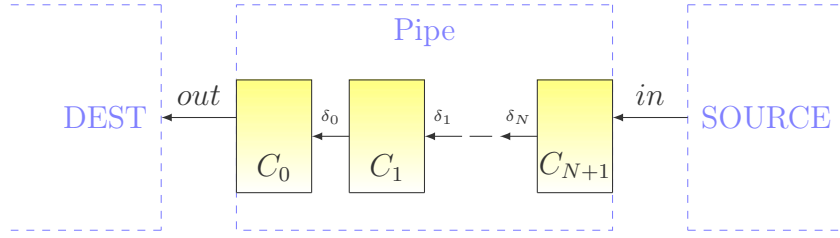
\begin{figure}
	\centering
\begin{tikzpicture}[only marks, scale=1]
\begin{scope}[>=latex]
   %draw dotted... 
   \draw[-] [blue!50, dashed] (0,0) -- (2,0);
   \draw[-] [blue!50, dashed] (2,0) -- node[left=1pt] {DEST} (2,-3);
   \draw[-] [blue!50, dashed] (2,-3) -- (0,-3);
   \draw[-] [blue!50, dashed] (8,0) -- (8,-3) -- (3,-3) -- (3,0);	
   \draw[-] [blue!50, dashed] (3,0) -- node[below=1pt] {Pipe} (8,0);   
%cell 1
   \draw[-] [top color = yellow!50, draw=black!80](3,-2.5) -- (3,-1) -- (4,-1)
-- (4,-2.5);	
   \draw[-] [black!80] (4,-2.5) -- node[above=0.5pt] {$C_0$} (3,-2.5);
%cell 2
   \draw[-] [top color = yellow!50, draw=black!80](4.5,-2.5) -- (4.5,-1) --
(5.5,-1) -- (5.5,-2.5);	
   \draw[-] [black!80](5.5,-2.5) -- node[above=0.5pt] {$C_1$} (4.5,-2.5);
%cell 3
   \draw[-] [top color = yellow!50, draw=black!80](7,-2.5) -- (7,-1) -- (8,-1)
-- (8,-2.5);	
   \draw[-] [black!80] (8,-2.5) -- node[above=0.5pt] {$C_{N + 1}$} (7,-2.5);   
   \draw[-] [blue!50, dashed] (11,0) -- (9,0);
   \draw[-] [blue!50, dashed] (9,0) -- node[right=1pt] {SOURCE} (9,-3);
   \draw[-] [blue!50, dashed] (9,-3) -- (11,-3);
   %draw arrows...
   \draw[->] [black!80] (3,-1.5) -> node[above=0.5pt] {$out$} (2,-1.5);
	 \draw[->] [black!80] (4.5,-1.75) -> node[above=0.5pt, font=\tiny]
{$\delta_0$} (4,-1.75);
	 \draw[->] [black!80] (6,-1.75) -> node[above=0.5pt, font=\tiny]
{$\delta_1$} (5.5,-1.75);
	 \draw[-] [black!80] (6.1,-1.75) -- (6.4,-1.75);
	 \draw[->] [black!80] (7,-1.75) -> node[above=0.5pt, font=\tiny]
{$\delta_N$} (6.5,-1.75);
   \draw[->] [black!80] (9,-1.5) -> node[above=0.5pt] {$in$} (8,-1.5);
\end{scope}[>=latex]
\end{tikzpicture}
\caption{The Software Architecture for $\pipe$}\label{pip}
\end{figure}

In Fig. \ref{buf} it has been assumed that $N$ cells are not connected
end-to-end but are used as a storage. These cells
interact with a centralised buffer controller which can store two more values
and uses the cells in
the storage as a circular queue (ordered as $0 < 1< ... < N-1$). In this case,
it is the buffer
controller that interacts with the external environment. More in detail, the
buffer controller
accepts a value from the external environment and then writes it in the first
empty cell. It also
reads the oldest undelivered value from the array and outputs it whenever
possible. 
In the following we write $a
\oplus b$ to denote $(a+b)\mbox{mod } N$.

\begin{definition}[{\it The buffer \buff}]
Let $i= 0 \ldots N-1$. The $i$-the cell of the storage is described by the
process $B_i \equiv
C[\Phi'_i]$ where the relabelling functions $\Phi'_i$ are defined by

%%%%%%%%%%%%%%%%%%%%%%%%%%%%
$$ \Phi' _{i}(\alpha)=
\begin{cases}
\begin{array}{ll}
\omega_{i} &  \mbox{if }  \alpha= in\\
\rho_{i}   & \mbox{if }  \alpha= out\\
\alpha     & \mbox{otherwise}\\
\end{array}
\end{cases}$$
%%%%%%%%%%%%%%%%%%%%%%%

Here we use the action $\omega_i$ ($\rho_i$) to denote the writing of a value
into the storage (the
reading of a value from the storage, respectively). Let $B = \{\omega_j ,\rho_j
\,|\, i=0,\ldots,
N\}$ be the set of all these actions and $\mem \equiv ( B_0
\parallel_{\emptyset} \ldots
\parallel_{\emptyset} B_{N-1})$.

The state of the buffer controller, $\bc(x,y,i,m)$, is determined by four
arguments: $x, y \in V=
\{\perp,\square\}$ are used to represent the absence or presence of an input
value (output value
resp.) (see below) stored in $\bc$, $i$ is the index of the cell that contains
the oldest
undelivered value and $m$ is the number of values currently stored in $\mem$.
If $x =\perp$ the buffer controller can accept a new value, otherwise (i.e.\ if
$x=\square$) it has
to wait until the last accepted value is actually stored in $\mem$. Analogously,
if $y=\square$,
then the buffer controller is ready to produce an output and if $y=\perp$ no
value is available for
immediate output.
Let $x, y \in V$, $0\leq i \leq N-1$ and $0\leq m\leq N$. We define:
\begin{enumerate}
% - 1
\item $\bc(\perp, \perp, i,0) \equiv in.\bc(\square, \perp, i, 0)$;
% - 2
\item $m >0$ implies $\bc(\perp, \perp, i,m) \equiv
in.\bc(\square, \perp, i,m) + \rho_i.\bc(\perp, \square, i\oplus 1, m-1)$;
% - 3
\item $\bc(\square, \perp, i,0) \equiv \omega_i.\bc(\perp,\perp, i,1)$;
% - 4
\item  $0<m<N$ implies $\bc(\square, \perp, i,m) \equiv
\omega_{i\oplus m}.\bc(\perp, \perp, i, m+1) +
\rho_{i}.\bc(\square,\square, i\oplus 1, m-1)$;
% - 5
\item   $\bc(\square, \perp, i,N)
\equiv \rho_{i}.\bc(\square,\square, i\oplus 1, N-1)$;
% - 6
\item   $\bc(\perp,\square, i,m)
\equiv in.\bc(\square,\square,i,m) + out.\bc(\perp,\perp,i,m)$;
% - 7
\item  $m<N$ implies $\bc(\square, \square, i,m) \equiv
\omega_{i\oplus m}.\bc(\perp, \square, i, m+1) +
out.\bc(\square,\perp, i, m)$;
% - 8
\item  $\bc(\square, \square, i,N)
\equiv out.\bc(\square,\perp, i, N)$.
\end{enumerate}

We finally define $\buff \equiv (\mem \,\|_B \, \bc(\perp, \perp, 0,0))/B$.
Notice that in such a
case all the actions we use to read and write values in $\mem$ are made
internal.
\end{definition}

\begin{proposition}\label{prop:buff}
For any $N\geq 1$, we have $rp_{\buff}(n) = 4n$.
\end{proposition}

\begin{proof}
Also in this case we have used \fase\ to prove that $\RRTS(\buff)$ does not have
catastrophic cycles
and to evaluate its asymptotic performance.
Concerning its response performance, consider first the case of one value: after
a time step, the
value is taken into the input part of $\bc$; after another time step, it is
moved into $\mem$; after
the third time step it is moved into the output part of $\bc$; after the fourth
time step, it is
delivered. For several values, these sequences can be interleaved to some
degree; but since $\bc$
takes part in each action, all these actions are performed sequentially, and
always after a time
step in the worst case.
E.g. for $n=kN + m$ for some $k \geq 1$ and $m < N$, 
first we fill up and clear the buffer with the sequence $((\A \, in \A \,
\tau)^{N}\, (\A \,\tau \A \, out)^{N})^{k}$, 
fill it up again with a sequence $(\A \, in \A \, \tau)^{m}$ and finally empty
it with the sequence $(\A \, \tau \A \, out)^{m-1} \, \A \, \tau \, \A $. 
All paths in that form (up to permutations) are $n$-critical paths with the
maximum number of time steps that is  $4Nk + 2m + 2(m - 1) + 2 = 4n$.
\end{proof} 

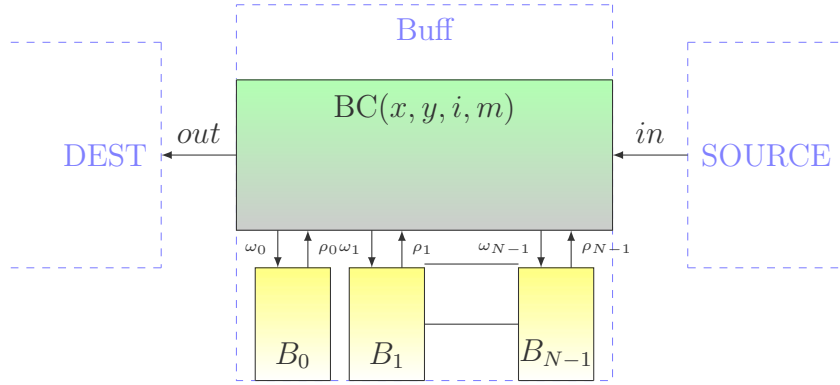
\begin{figure}
	\centering
	\begin{tikzpicture}[only marks, scale=1]
		\begin{scope}[>=latex]
		   %draw dotted... 
		   \draw[-] [blue!50, dashed] (0,0) -- (2,0);
		   \draw[-] [blue!50, dashed] (2,0) -- node[left=1pt] {DEST}
(2,-3);
		   \draw[-] [blue!50, dashed] (2,-3) -- (0,-3);
		   \draw[-] [blue!50, dashed] (8,0.5) -- (8,-4.5) -- (3,-4.5) --
(3,0.5);	
		   \draw[-] [blue!50, dashed] (3,0.5) -- node[below=1pt] {Buff}
(8,0.5);  
		   \draw[-] [top color = green!30, bottom color = black!20,
draw=black!80](8,-0.5) -- (8,-2.5) -- (3,-2.5) -- (3,-0.5);	
		   \draw[-] [black!80] (3,-0.5) -- node[below=1pt]
{BC($x,y,i,m$)} (8,-0.5); 
				%cell 1
		   \draw[-] [top color = yellow!50, draw=black!80](3.25,-4.5) --
(3.25,-3) -- (4.25,-3) -- (4.25,-4.5);	
		   \draw[-] [black!80] (4.25,-4.5) -- node[above=0.5pt] {$B_0$}
(3.25,-4.5);
				%cell 2
		   \draw[-] [top color = yellow!50, draw=black!80](4.5,-4.5) --
(4.5,-3) -- (5.5,-3) -- (5.5,-4.5);	
		   \draw[-] [black!80] (5.5,-4.5) -- node[above=0.5pt] {$B_1$}
(4.5,-4.5);
				%cell 3
		   \draw[-] [top color = yellow!50, draw=black!80](6.75,-4.5) --
(6.75,-3) -- (7.75,-3) -- (7.75,-4.5);	
		   \draw[-] [black!80] (7.75,-4.5) -- node[above=0.5pt] {$B_{N -
1}$} (6.75,-4.5);      
		   \draw[-] [blue!50, dashed] (11,0) -- (9,0);
		   \draw[-] [blue!50, dashed] (9,0) -- node[right=1pt] {SOURCE}
(9,-3);
		   \draw[-] [blue!50, dashed] (9,-3) -- (11,-3);
				%draw arrows... 
			 \draw[->] [black!80] (3.55,-2.5) -> node[left=-0.2pt,
font=\tiny] {$\omega_0$} (3.55,-3);
			 \draw[->] [black!80] (3.95,-3) -> node[right=-0.2pt,
font=\tiny] {$\rho_0$} (3.95,-2.5);
			 \draw[->] [black!80] (4.8,-2.5) -> node[left=-0.2pt,
font=\tiny] {$\omega_1$} (4.8,-3);
			 \draw[->] [black!80] (5.2,-3) -> node[right=-0.2pt,
font=\tiny] {$\rho_1$} (5.2,-2.5);
			 \draw[->] [black!80] (7.05,-2.5) -> node[left=-0.2pt,
font=\tiny] {$\omega_{N-1}$} (7.05,-3);
			 \draw[->] [black!80] (7.45,-3) -> node[right=-0.2pt,
font=\tiny] {$\rho_{N-1}$} (7.45,-2.5);
			 \draw[-] [black!80] (5.5,-3.75) -- (6.75,-3.75);
			 \draw[-] [black!80] (5.5,-2.95) -- (6.75,-2.95);
		   \draw[->] [black!80] (9,-1.5) -> node[above=0.5pt] {$in$}
(8,-1.5);
		   \draw[->] [black!80] (3,-1.5) -> node[above=0.5pt] {$out$}
(2,-1.5);   
		\end{scope}[>=latex]   
	\end{tikzpicture}
	\caption{The Software Architecture for $\buff$}\label{buf}
\end{figure}

Now we can state the main result of this paper. This follows as a
straightforward consequence of
Propositions~\ref{prop:fifo}, \ref{prop:pipe} and~\ref{prop:buff}.

\begin{corollary}
For any $N\geq 1$ , \fifo\ is more efficient than  both \pipe\ and \buff
(w.r.t. the quantitative point of view). Moreover, \buff\ is more
efficient than \pipe\ iff $n \leq \lfloor \dfrac{N+1}{2} \rfloor$.
\end{corollary}

\section{Concluding remarks}
\label{sec:concluding}
The results obtained with our tool are quite different from those presented
in~\cite{CDV01} where
the same buffer implementations have been compared using the efficiency preorder
defined
in~\cite{CVJ02}. In~\cite{CDV01} it is stated that \fifo\ and \pipe\ are
unrelated according to the
worst-case efficiency preorder (unrelated means that the former process is not
more efficient than
the second one and vice versa). Similarly $\buff$ and \pipe\ are unrelated.
The authors provide good reasons for these results and also prove that \fifo\ is
more efficient than \buff\ but not vice versa.

As already stated in the introduction, the efficiency
preorder is based on arbitrary test environments, whereas we have only used
restricted environments
adequate for quantitative reasoning in this paper. To explain the results
of~\cite{CDV01}, we
consider the refusal trace $v=in \, \A \, \emptyset \, out\, \{in\} \in
\RT(\fifo) \backslash
\RT(\pipe)$, which can be understood as a witness of slow behaviour of $\fifo$,
justifying $\fifo
\not\sqsupseteq \pipe$. This trace tells us that $\fifo$ can perform two time
steps after an $in$
provided the environment does not offer a communication after the first one
($\fifo$ itself would
neither block $in$ nor $out$); then it can deliver the value and can now delay
$in$ (as after any
visible action).
Now we show that none of our users can be such a suitable context, i.e. that
\fifo\ cannot
participate in such a discrete trace $v$ when running in parallel with a user
$U_n$; hence, $v$
is not relevant for $rp_{\fifo}$.

$\fifo \,\|\, U_n \nar{in}{r} \fifo(1) \,\|\, (U_{n-1}
\,\|_{\{\omega\}} \underline{out}.\underline{\omega})
\nar{\A}{r} P' = \underline{\fifo}(1) \,\|\,(U_{n-1}
\,\|_{\{\omega\}} \underline{out}.\underline{\omega})$

Here, $\underline{\fifo}(1) = (\underline{in}.\fifo(2) +
\underline{out}.\fifo(0))$ can perform $\nar{\emptyset}{r}$ to itself;
but by the refusal semantics we could have $P' \nar{1}{}$ only if
$(U_{n-1} \,\|_{\{\omega\}} \underline{out}.\underline{\omega})$ is
able to refuse both $in$ and $out$. And this is clearly not the case. We are
currently working on this qualitative/quantitative issue by defining a slight
variation of the faster than preorder as given in \cite{CVJ02} to relate
processes w.r.t. the restricted class of tests ${\cal U}$ as in \cite{CV05} but
by some variant of refusal trace inclusion.

Our aim is to tune  $\fase$  to allow the analysis of larger systems, where the
performance module needs more attention since it implements the theories
introduced above. A first important result, we have already obtained, is the
improvement of the catastrophic-cycles detection; ensuring their absence is the
basis for any further performance analysis. 
A second result regards the calculation of the bad cycle, especially when we
consider complex processes. However, the graph $G'$ used in Karp's algorithm
could be very large, and we will investigate ways to minimise it. 
We are also working on a good strategy to determine the response performance of
$P$ for a given $n$. Different approaches are under investigation but they still
need to be validated. Currently, $\fase$ executes an exhaustive search on
$\RRTS(P)$ that looks for the $n$-critical path whose duration is maximal;
clearly as $n$ increases this solution becomes soon intractable, especially for
complex processes. Even though it is a rough solution, at least it helped to
validate the results on response performance presented in the above
propositions.

Anyhow, $\fase$ represents a good first step towards the creation of an
integrated framework for the analysis of concurrent systems modelled through
PAFAS. The improvements introduced with $\fase$ and the possibility to derive
the complete set of behavioural traces of the modelled system allowed us to
study and validate many results, such as the ones stated in this paper, that
would have been harder to calculate without an automated tool like $\fase$.
%
% ---- Bibliography ----
%

%
\end{document}